\documentclass[11pt,a4paper]{article}
\usepackage{amsmath}
\usepackage[T1]{fontenc}
\usepackage[utf8]{inputenc}
\usepackage{amsfonts}
\usepackage{amsthm}
\usepackage{amstext}
\usepackage{amssymb}
\usepackage{color}
\usepackage{bbold}
\usepackage{hyperref}
\usepackage{url}
\newcommand{\rr}{\mathbb{R}}
\newcommand{\qq}{\mathbb{Q}}

\newcommand{\sharpp}{\sharp \sf{P}}

\newcommand{\Ind}{\mathrm{Ind}}
\newcommand{\rk}{\mathrm{rank}}
\newcommand{\tr}{\mathrm{Tr}}

\newtheorem{theorem}{Theorem}
\newtheorem{example}{Example}

\newtheorem{lemma}[theorem]{Lemma}
\newtheorem*{lemma*}{Lemma}
\newtheorem{proposition}[theorem]{Proposition}
\newtheorem{corollary}[theorem]{Corollary}
\newtheorem{remark}[theorem]{Remark}

\newtheorem*{open*}{Open~question}

\begin{document}

\title{On the complexity of partial derivatives}

\author{Ignacio Garcia-Marco, Pascal Koiran, Timoth\'ee Pecatte, 
Stéphan Thomassé\\
LIP\thanks{UMR 5668 Ecole Normale Supérieure de  Lyon, CNRS, UCBL, INRIA.
Email:  
{\tt [Pascal.Koiran, Timothee.Pecatte, Stephan.Thomasse]@ens-lyon.fr,} 
{\tt iggarcia@ull.es}.
The authors are supported by ANR project
CompA (code ANR--13--BS02--0001--01).
}, 
Ecole Normale Sup\'erieure de Lyon, Universit\'e de Lyon.
}

\maketitle

\begin{abstract}
The method of partial derivatives is one of the most successful lower bound 
methods for arithmetic circuits. It uses as a complexity measure the dimension 
of the span of the partial derivatives of a polynomial.
In this paper, we consider this complexity measure as a computational problem:
for an input polynomial given as the sum of its nonzero monomials,
what is the complexity of computing the dimension of its space 
of partial derivatives? 

We show that this problem is $\sharpp$-hard 
and we ask whether it belongs to  $\sharpp$. 
We analyze the ``trace method'', recently used in combinatorics and in algebraic complexity to lower bound the rank of certain matrices. 
We show that this method provides a 
polynomial-time computable lower bound on the dimension of the span
of partial derivatives, and from this method we derive closed-form
lower bounds. We leave as an open problem the existence of an approximation 
algorithm with reasonable performance guarantees.
\end{abstract}

\section{Introduction}

Circuit lower bounds against a class of circuits $\cal C$ are often obtained 
by defining an appropriate complexity measure which is small
for small circuits of $\cal C$ but is high for some explicit `` hard function.''
For arithmetic circuits, one of the most successful complexity measures 
is based on partial derivatives.
Sums of powers of linear forms provide the simplest model where the
method of partial derivatives can be presented (see for instance Chapter~10 of 
the survey by Chen, Kayal and Wigderson~\cite{CKW11}).
In this  model, 
a homogeneous polynomial $f(x_1,\ldots,x_n)$ 
of degree $d$ is given by an expression of the form:
\begin{equation} \label{waring}
f(x_1,\ldots,x_n)=\sum_{i=1}^r l_i(x_1,\ldots,x_n)^d
\end{equation}
where the $l_i$'s are linear functions. The smallest possible $r$ 
is often called the Waring rank of $f$ in the algebra litterature.
One takes as complexity measure $\dim \partial^{=k} f$, 
where $\partial^{=k} f$ denotes the linear space of polynomials spanned 
by the partial derivatives of $f$ of order $k$.
For any $k \leq d$, 
the derivatives of order $k$ of a $d$-th power of a linear form 
$l(x_1,\ldots,x_n)$ are constant multiples of $l^{d-k}$.
Therefore, by linearity of derivatives we have for any $k \leq d$ 
the lower bound 
$r \geq \dim \partial^{=k} f$ on the Waring rank of $f$.

The method of partial derivatives was introduced in the complexity theory 
litterature by Nisan and Wigderson~\cite{NW96}, 
where lower bounds were given for more powerful models than~(\ref{waring}) 
such as e.g. depth 3 arithmetic circuits. In such a circuit, the $d$-th powers 
in~(\ref{waring}) are replaced by products of $d$ affine functions. 
We then have~\cite{NW96} the lower bound $r \geq (\dim \partial^* f)/2^d$, 
where $r$ denotes as in~(\ref{waring}) the fan-in of the circuit's output gate
and $\partial^* f$ denotes the space spanned by partial derivatives of 
all order. 
More recently, a number of new lower bound results were obtained using
a refinement of the method of partial derivatives. These new results are based
on ``shifted partial derivatives'' (see the continuously updated online survey
maintained by Saptharishi~\cite{Saptsurvey} 
for an extensive list of references), but we will 
stick to ``unshifted'' derivatives in this paper.

Partial derivatives can also be used for {\em upper bound} results: see in particular Theorem~5 in~\cite{kayal2012affine} for an algorithm 
that constructs a representation in the Waring model~(\ref{waring}) 
of a polynomial given by a black box. To learn more on the complexity of 
circuit reconstruction for various classes of arithmetic circuits one may 
consult Chapter~5 of the survey by Shpilka and Yehudayoff~\cite{SY10}.

\subsection*{Our contributions}

In this paper we consider the dimension of the set of partial derivatives
as a computational problem and provide the first results (that we are aware of) on its complexity. This is quite a natural problem since, as explained above,
the knowledge of this dimension for an input polynomial $f$ provides estimates
on the circuit size of $f$ for several classes of arithmetic circuits.
We assume that the input polynomial $f$ is given in
the sparse representation  
(also called ``expanded representation''), i.e., as the 
sum of its nonzero monomials. 
We show in Section~\ref{hardsection} that computing $\dim \partial^*f$ 
is hard for Valiant's~\cite{valiant79a} counting class $\sharpp$. 
This remains true even if $f$ is multilinear, 
homogeneous and has only 0/1 coefficients.
The precise complexity of this problem remains open, in particular 
we do not know whether computing $\dim \partial^*f$ is in $\sharpp$.

As an intermediate step toward our $\sharpp$-hardness result, 
we obtain a result of independent interest for a problem of topological
origin: computing the number of faces in an abstract simplicial complex.
Our $\sharpp$-hardness proof for this problem proceeds by reduction 
from counting 
the number of independent sets in a graph, 
a well-known $\sharpp$-complete problem~\cite{provan83}.
It is inspired by the  recent proof~\cite{roune13} that computing the Euler characteristic of abstract simplicial complexes is $\sharpp$-complete.

Since the $\sharpp$-hardness result rules out an efficient algorithm for
the exact computation of $\dim \partial^*f$, it is of interest to obtain
efficiently computable upper and lower bounds for this quantity and for
$\dim \partial^{=k} f$. Upper bounds are easily obtained from the linearity
of derivatives. In Section~\ref{upperlower} we give a lower bound that is 
based on the consideration of a single ``extremal'' monomial of $f$.
In particular, for a multilinear homogeneous polynomial of degree $d$ 
with $s$ monomials 
we have ${d \choose k} \leq \dim \partial^{=k} f \leq s{d \choose k}$ 
for every $k$.
In Section~\ref{tracesec} we provide lower bounds that take all monomials 
of $f$ into account. Depending on the choice of the input polynomial,
these lower bounds may be better or worse than the lower
bound of Section~\ref{upperlower}. The lower bounds of Section~\ref{tracesec}
are based on the ``trace method.'' This method was recently used in~\cite{kayal2014exponential,kayal15} to lower bound 
the dimension of {\em shifted} partial derivatives 
of a specific ``hard'' polynomial, the so-called Nisan-Wigderson polynomial.
In~\cite{kayal2014exponential} this method is attributed to Noga Alon~\cite{Alon09}. 

In a nutshell, the principle of the trace method is as follows. Suppose
that we want to lower bound the rank of a matrix $M$. In this paper, $M$ will be the matrix of partial derivatives of a polynomial $f(x_1,\ldots,x_n)$. From $M$, we construct the symmetric matrix $B=M^T.M$. We have $\rk(M) \geq \rk(B)$, 
with equality if the ranks are computed over the field of real numbers.
In the trace method, we replace $\rk(M)=\rk(B)$ by the ``proxy rank'' 
$\tr(B)^2/ \tr(B^2)$. This is legitimate due to the the inequality 
\begin{equation}
\rk(B) \geq \tr(B)^2/ \tr(B^2),
\end{equation}
which follows from the Cauchy-Schwarz inequality applied to the eigenvalues of $B$. It is often easier to lower bound the proxy rank than to lower bound the rank directly. 
In Section~\ref{tracesec} we will see that the proxy rank can be computed in
polynomial time. This is not self-evident because 
$B$ may be of size exponential in the number $n$ of variables of $f$.
By contrast, as explained above computing $\rk(B)$ over the field of real numbers is $\sharpp$-hard.

\subsection*{Organization of the paper}

In Section~\ref{upperlower} we set up the notation for the rest of the paper,
and give some elementary estimates. In particular, Theorem~\ref{lexico}
provides a lower bound that relies on the consideration of a single extremal
monomial of $f$. 
Section~\ref{tracesec} is devoted to the trace method.
We use this method to derive closed-form lower bounds on the dimension
of the space of partial derivatives, and compare them to the lower bound
from Theorem~\ref{lexico}. In Section~\ref{polytime} we show that the ``proxy rank'' $\tr(B)^2/\tr(B^2)$ is computable in polynomial time.
In Section~\ref{elementarypoly} we show that the trace method behaves very 
poorly on elementary symmetric polynomials: for certain settings of parameters, 
the matrix of partial derivatives has full rank but the trace method can only
show that its rank is larger than 1.
Finally, we show in Section~\ref{hardsection} that it is $\sharpp$-hard 
to compute $\dim \partial^* f$ and to compute the number of faces in an abstract simplical complex.

\subsection*{Open problems}

Here are three of the main problems that are left open by this work.
\begin{enumerate}
\item Give a nontrivial upper bound on the complexity of computing $\dim \partial^* f$ and $\dim \partial^{=k} f$. 
In particular, are these two problems in $\sharpp$?

\item Give an efficient algorithm that approximates $\dim \partial^* f$ or $\dim \partial^{=k} f$, and comes with a reasonable performance guarantee (or show
that such an algorithm does not exist).
The proxy rank $\tr(B)^2/\tr(B^2)$ is efficiently computable, but certainly
does not fit the bill due to its poor performance on symmetric polynomials.
For counting the number of independent sets in a graph (the starting point of our reductions), there is already a significant amount of work on approximation
algorithms~\cite{luby97,dyer00} and hardness of approximation~\cite{luby97,dyer02}.

\item We recalled at the beginning of the introduction that partial derivatives 
are useful as a complexity measure to prove lower bounds against several classes
of arithmetic circuits. We saw that computing this measure is hard, 
but is it hard to compute the Waring rank of a homogeneous polynomial $f$ 
given in expanded form, or to compute the size of the smallest (homogeneous) depth 3 circuit for $f$?
The former problem has been recently proved to be NP-hard already for polynomials of degree~3 \cite{Shitov16}.
\end{enumerate}

\section{Elementary bounds} \label{upperlower}

We use the notation $\partial_{\beta} f$ for partial derivatives of 
a polynomial $f(x_1,\ldots,x_n)$. Here $\beta$ is a $n$-tuple of integers,
and $\beta_i$ is the number of times that we differentiate $f$ with respect
to $x_i$. 
We denote by $\partial^{=k} f$ the linear space spanned 
by the partial derivatives of $f$ of order $k$, 
and by $\partial^* f$ the space spanned by partial derivatives of all order.
For $\alpha \in \{0,1\}^n$, we denote by $x^\alpha$ 
the multilinear  monomial $x_1^{\alpha_1}.\cdots.x_n^{\alpha_n}$.
More generally, if $\alpha$ is a $n$-tuple of integers,
$x^\alpha$ denotes the monomial 
$x_1^{\alpha_1}.\cdots.x_n^{\alpha_n}/(\alpha_1!\cdots\alpha_n!)$.
These monomials form a basis of the space $\rr[x_1,\ldots,x_n]$
of real polynomials in $n$ variables, which we refer to as the
``scaled monomial basis.''
Dividing by the constant $\alpha_1!\cdots\alpha_n!$ is convenient since 
differentiation takes the simple form: 
$\partial_{\beta} x^{\alpha} = x^{\alpha-\beta}$.
We agree that $x^{\alpha-\beta}=0$ if one of the components of $\alpha-\beta$ 
is negative.

For a monomial $f=x^{\alpha}$, 
$\dim \partial^* x^\alpha = \prod_{i=1}^n (\alpha_i+1)$.
One can compute $\dim \partial^{=k} x^{\alpha}$ by dynamic programming thanks to the
recurrence relation:
$$\dim \partial^{=k} x^{\alpha} = 
\sum_{j=0}^{\alpha_1} \dim \partial^{=k-j} (x_2^{\alpha_2}.\cdots.x_{n}^{\alpha_{n}}).$$
It takes altogether $O((\deg f)^2)$ additions 
to compute the $\deg(f)+1$ numbers $\dim \partial^{=k} f$ 
for $k=0,\ldots,\deg(f)$.
Equivalently, $\dim \partial^{=k} x^{\alpha}$ can be computed as 
the coefficient of $t^k$ in the polynomial 
$$(1+t+\ldots+t^{\alpha_1}).(1+t+\ldots+t^{\alpha_2}). \cdots . 
(1+t+\ldots+t^{\alpha_n}).$$

For a polynomial with more than one monomial, one can obtain simple upper bounds
thanks to the linearity of derivatives  since 
$\dim \partial^* (f+g) \leq \dim \partial^* f + \dim \partial^* g$
and $\dim \partial^{=k} (f+g) \leq \dim \partial^{=k} f + \dim \partial^{=k} g$.
Lower bounding the dimension of the space of partial derivatives is slightly 
less 
immediate. 
\begin{theorem} \label{lexico}
For any polynomial $f$ there is a monomial $m$ in $f$ such that 
$\dim \partial^{=k} f \geq \dim \partial^{=k} m$ for every $k$. 
In particular, if all monomials 
in $f$ contain at least $r$ variables 
then 
$\dim \partial^{=k} f \geq {r \choose k}$ for every $k$.
\end{theorem}
\begin{proof}
The second claim clearly follows from the first claim. 
Let $n$ be the number of variables in $f$. In order to find 
the monomial $m$, we fix a total order $\leq$ on $n$-tuple of integers
which is compatible with addition, for instance the lexicographic order
(different orders may lead to different $m$'s).
We will use $\leq$ to order monomials as well as tuples $\beta$ in partial 
derivatives such as $\partial_{\beta} f$.
We will also use the partial order $\subseteq$ defined by:
$\beta \subseteq \alpha$ iff $\beta_i \leq \alpha_i$ for all $i=1,\ldots,n$.
Let $m=x^{\alpha}$ be the smallest monomial for $\leq$ with a nonzero coefficient in $f$.

To complete the proof of the theorem, we just need 
to  show that the partial derivatives $\partial_{\beta} f$ 
where $\beta \subseteq \alpha$ are linearly independent.
The dimension of the space spanned by these partial derivatives is equal 
to the rank of a certain matrix $M$. The rows of $M$ are indexed by
the $n$-tuples $\beta$ such that $\beta \subseteq \alpha$, and row $\beta$
contains the coordinates of $\partial_{\beta} f$ 
in the scaled monomial basis $(x^{\gamma})$. 
If $f=\sum_{\gamma} a_{\gamma} x^{\gamma}$, we therefore have 
$M_{\beta, \gamma-\beta}=a_{\gamma}$.
Let us order the rows and columns of $M$ according to $\leq$.
We have seen that $M$ contains a nonzero coefficient in row $\beta$ 
and column $\alpha-\beta$. 
This coefficient is strictly to the left of any nonzero coefficient in any row
above $\beta$. Indeed, we have $\alpha-\beta < \alpha'-\beta'$ 
if $\alpha' \geq \alpha$ and $\beta' < \beta$. 
Our matrix is therefore in row echelon form, and does not contain any identically zero row. It is therefore of full row rank.
\end{proof}
\begin{remark} \label{newton}
Recall that the Newton polytope of $f$ is the convex hull of the $n$-tuples of
exponents of monomials of $f$. By changing the order $\leq$ in the proof of
Theorem~\ref{lexico} we can take for $m$ any vertex of the Newton polytope.
\end{remark}
Theorem~\ref{lexico} lower bounds $\dim \partial^{=k} f$ by the same 
dimension computed for a suitable monomial of $f$.
 This is of course tight if $f$ has
a single monomial. We note that adding more monomials does not necessarily increase
$\dim \partial^{=k} f$. For instance, the polynomial 
$f=\prod_{i=1}^d\sum_{j=1}^q x_{ij}$
has $q^d$ monomials but $\dim \partial^{=k} f$ remains equal to ${d \choose k}$ 
for any $q$.
\begin{corollary}
For a multilinear homogeneous polynomial of degree $d$ with $s$ monomials 
we have ${d \choose k} \leq \dim \partial^{=k} f \leq s{d \choose k}$ 
for every $k$.
\end{corollary}
\begin{proof}
The upper bound follows from the linearity of derivatives, and the lower bound 
from Theorem~\ref{lexico}.
\end{proof}

\section{The trace method} \label{tracesec}

The lower bound on $\dim \partial^{=k} f$ in Theorem~\ref{lexico} takes  a single monomial of $f$ into account. 
In this section we give a more ``global'' 
result which takes all monomials into account.
We will in fact lower bound the dimension of a subspace of $\partial^{=k} f$,
spanned by partial derivatives of the form $\partial_I f$ 
where $I \in \{0,1\}^n$. In other words, we will differentiate at most once 
with respect  to any variable.\footnote{One could lift this restriction and 
derive similar results for the ``full'' matrix of $k$-th order derivatives, i.e., for the case where several differentiations with respect to the same variable are allowed. 
This would have the effect of replacing the binomial coefficients ${\sup(P) \choose k}$ in the lower bounds of the present section by $\dim \partial^{=k}P$. Here $P$ denotes a monomial of $f$; we have explained at the beginning of Section~\ref{upperlower} 
how to compute $\dim \partial^{=k}P$.

We will stick here to a single differentiation for the sake of notational simplicity.} 
We can of course view $I$ as a subset of $[n]$ 
rather than as a vector in $\{0,1\}^n$.

We form a matrix $M$ of partial derivatives as in the proof 
of Theorem~\ref{lexico}. The rows of $M$ are indexed by subsets 
of $[n]$ of size $k$, and row $I$ contains the expansion 
of $\partial_I f$ in the basis  $(x^J)$.
If $f = \sum_j a_J x^J$, we have seen in Section~\ref{upperlower} 
that $M_{I,J}=a_{I+J}$.
In order to lower bound the rank of $M$, we will apply the following lemma 
to the symmetric matrix $B=M^T.M$.
\begin{lemma} \label{alon}
For any real symmetric 
matrix  $B \neq 0$ we have
$$\rk(B) \geq \frac{(\tr B)^2}{\tr(B^2)}.$$
\end{lemma}
Lemma~\ref{alon} is easily obtained by applying the Cauchy-Schwarz inequality 
to the vector of nonzero eingenvalues of $B$.
Note that $B=M^T.M$ has same rank as $M$ since we have: 
$x^TBx=0 \Leftrightarrow Mx=0$ for any vector $x$.

We first consider the case of polynomials with 0/1 coefficients, for which
 we have the following lower bound.
\begin{theorem} \label{nonmult}
For $f$ a real polynomial with 0/1 coefficients we have
 \begin{equation} \label{01lb}
\dim \partial^{=k} f \geq \frac{\sum_{P \in {\cal M}} {\sup(P) \choose k}}{|{\cal M}|^2}
\end{equation}
where $\cal M$ denotes the set of monomials occuring in $f$, and $\sup(P)$ 
the number of distinct variables occuring in monomial $P$.
\end{theorem}
The right-hand side of~(\ref{01lb}) is sandwiched between ${\mathrm{supmin} \choose k}/|{\cal M}|$ and ${\mathrm{supmax} \choose k}/|{\cal M}|$, where $\mathrm{supmin}$ and $\mathrm{supmax}$ denote respectively the minimum and maximum number of variables occuring in a monomial of $f$.
Theorem~\ref{lexico} provides a better lower bound (by a factor of $|{\cal M}|$)
when all the monomials of $f$ have supports of same size.
Theorem~\ref{nonmult} becomes interesting when all but a few monomials in 
$f$ have large support. Indeed, the presence of a few monomials of small support can ruin the lower bound of Theorem~\ref{lexico}.
\begin{example}
Let $f(x_1,\ldots,x_n)=x_1.x_2.\cdots.x_n+\sum_{i=1}^n x_i^n$.
The Newton polytope of $f$ is an $n$-simplex whose vertices correspond to the monomials $x_1^n,\ldots,x_n^n$. 
The point corresponding to the monomial $x_1.x_2.\cdots.x_n$ is the barycenter
of this simplex, and in particular it is not a vertex of the Newton polytope.
As a result, by Remark~\ref{newton} the lower bound method of Theorem~\ref{lexico} can only  show that $\dim \partial^{=k} f \geq 1$.
Theorem~\ref{nonmult} shows the better lower bound:
$$ \dim \partial^{=k} f \geq \frac{{n \choose k}+n}{(n+1)^2}.$$

It is not hard to check by a direct calculation that for this example, 
the correct value of $\dim \partial^{=k} f$ is:
\begin{itemize}
\item 1 for $k  \in \{0,n\};$
\item $n$ for $k  \in \{1,n-1\};$
\item ${n \choose k} + n$ for $2 \leq k \leq n-2$.
\end{itemize}
\end{example}

Let us now proceeed with the proof of Theorem~\ref{nonmult}.
In view of Lemma~\ref{alon}, we need a lower bound on $\tr(B)$ and an upper bound on $\tr(B^2)$.
\begin{lemma} \label{trbnonmult}
$\tr(B)=\sum_{P \in {\cal M}} {\sup(P) \choose k}.$
\end{lemma}
\begin{proof}
By definition of $B$, $\tr(B)=\sum_{J} (M^T.M)_{J,J}=\sum_{I,J}M_{I,J}^2$. Since $M_{I,J} \in \{0,1\}$, this is nothing but the number of nonzero entries in $M$. 
Monomial $P$ contributes ${\sup(P) \choose k}$ such entries and they are all
distinct.
\end{proof}

\begin{lemma} \label{trb2nonmult}
$\tr(B^2) \leq |{\cal M}|^2\sum_{P \in {\cal M}} {\sup(P) \choose k}.$
\end{lemma}
\begin{proof}
Since $B$ is symmetric, $\tr(B^2)=\sum_{K,L} (B_{K,L})^2$. 
By definition of $B$, $B_{K,L}=\sum_I M_{I,K}.M_{I,L}$.
Therefore $$\tr(B^2)=\sum_{I,J,K,L} M_{I,K}.M_{I,L}.M_{J,K}.M_{J,L}.$$
In this formula, $I,J$ range over row indices (subsets of $[n]$ of size $k$), and $K,L$ range over column indices. 
Hence $\tr(B^2)$ is equal to the number of quadruples $(I,J,K,L)$ such that all 4 entries $M_{I,K}$, $M_{I,L}$, $M_{J,K}$, $M_{J,L}$ are nonzero. Let us say that a quadruple is {\em valid} if this condition is satisfied. 
A quadruple is valid if and only if the 4 coefficients $a_{I+K}$, $a_{I+L}$, 
$a_{J+K}$, $a_{J+L}$ are nonzero. This implies that there are at most 
$|{\cal M}|^2\sum_{P \in {\cal M}} {\sup(P) \choose k}$
valid quadruples. Let us indeed denote by $P,Q,R$ the 3 $n$-tuples
 ${I + K}$, ${I + L}$, ${J + K}$.
For every fixed $P$ we have at most 
${\sup(P) \choose k}$ choices for $I$ since $I$ is contained in the support of $P$, and at most ${|\cal M|}^2$ choices for the pair $(Q,R)$. 
The result follows since the quadruple $(I,J,K,L)$ 
is completely determined by the 
choices of $P,Q,R$ and $I$: we must have $K=P - I$,
$L = Q -I$, $J=R-K$.
\end{proof}
Theorem~\ref{nonmult} follows immediately from Lemmas~\ref{alon} to~\ref{trb2nonmult}.

\subsection{Extension to real coefficients}

In this section we generalize Theorem~\ref{nonmult} to polynomials 
with real coefficients. Theorem~\ref{realth} could itself be generalized 
to polynomials with complex coefficients by working with a Hermitian matrix 
in Lemma~\ref{alon} rather than with a symmetric matrix.
\begin{theorem} \label{realth}
For any real polynomial $f$ we have
 \begin{equation} \label{reallb}
\dim \partial^{=k} f \geq \frac{\sum_{P \in {\cal M}} {\sup(P) \choose k}a_P^2}{|{\cal M}|\sum_{P \in {\cal M}} a_P^2}
\end{equation}
where $\cal M$ denotes the set of monomials occuring in $f$.
\end{theorem}
To 
make sense of the lower bound in this theorem, it is helpful to look at a 
couple of  special cases. 
If the coefficients $a_P$ all have the same absolute value,
e.g., $|a_P|=1$ for all $P \in {\cal M}$, the right-hand side of~(\ref{reallb}) reduces to 
${\sum_{P \in {\cal M}} {\sup(P) \choose k}}/{|{\cal M}|^2}$.
This is exactly the lower bound in Theorem~\ref{nonmult} (but now the coefficients of $f$ may be in $\{-1,0,1\}$ rather than $\{0,1\}$).

Our lower bound becomes weaker when the vectors 
$({\sup(P) \choose k})_{P \in {\cal M}}$ and $(a_P^2)_{P \in {\cal M}}$
are approximately orthogonal. This can happen when the monomials with large
support have small coefficients. In this case, as should be expected, 
Lemma~\ref{alon} is effectively unable to detect the presence of monomials
of large support. 
A probabilistic analysis shows that this bad behavior is atypical. 
Consider for instance the following semirandom model: we first choose a 
set ${\cal M}$ of monomials in some arbitrary (worst case) way, and then 
the $a_P$ are drawn independently at random from some common
probability 
distribution such that $\Pr[a_P=0]=0$.
\begin{corollary} \label{realcol}
Let $\displaystyle L(f)=\frac{\sum_{P \in {\cal M}} {\sup(P) \choose k}a_P^2}{|{\cal M}|\sum_{P \in {\cal M}} a_P^2}$
be the lower bound on the right-hand side of~(\ref{reallb}).

In the semirandom model described above, the expectation of $L(f)$ is:
$$E[L(f)]=\sum_{P \in {\cal M}} {\sup(P) \choose k}/{|{\cal M}|^2}.$$
\end{corollary}
Note that we obtain for $E[L(f)]$ the lower bound from the case 
where $|a_P|=1$ for all $P \in \cal M$.
\begin{proof}[Proof of Corollary~\ref{realcol}]
We write $L(f)=\sum_{P \in {\cal M}} {\sup(P) \choose k}X_P/|{\cal M}|$ 
where $X_P$ is the random variable:
$$X_P=\frac{a_P^2}{\sum_{J \in {\cal M}} a_J^2}.$$
By linearity of expectation, 
$E[L(f)]=\sum_{P \in {\cal M}} {\sup(P) \choose k} E[X_P]/|{\cal M}|$.
From the i.i.d assumption all the $X_P$ have the same expectation, and
since $\sum_{P \in {\cal M}} X_P =1$ 
this common expectation must be $1/|{\cal M}|$.
\end{proof}
The remainder of this section 
is devoted to the proof of Theorem~\ref{realth}.
We follow the proof of Theorem~\ref{nonmult}.
In particular, we still differentiate at most once with respect 
to each variable, we define the same matrix $M$ of partial derivatives
and the symmetric matrix $B=M^T.M$. 
We again have $\dim \partial^{=k}f \geq \rk(M)=\rk(B)$;
hence Theorem~\ref{realth} follows from Lemma~\ref{alon} and from the next 
two lemmas.
\begin{lemma} \label{trbreal}
$\tr(B)=\sum_{P \in {\cal M}} {\sup(P) \choose k}a_P^2.$
\end{lemma}
\begin{proof}
By the proof of Lemma~\ref{trbnonmult}, $\tr(B)$ is equal to the sum of squared entries of $M$; and we have $M_{I,P-I}=a_P$ for each set $I$ of size $k$ contained in the support of $P$.
\end{proof}

\begin{lemma} \label{trb2real}
$\tr(B^2) \leq |{\cal M}|\tr(B)
\left(\sum_{R \in {\cal M}} a_R^2\right).$
\end{lemma}
\begin{proof}
By the proof of Lemma~\ref{trb2nonmult}, 
\begin{equation} \label{sumoverV}
\tr(B^2)=\sum_{(I,J,K,L)} a_{I+K}.a_{I+L}.a_{J+K}.a_{J+L}.
\end{equation}
Here $I,J$ range over row indices of $M$ while $K,L$ range over column indices.
As in the proof of Lemma~\ref{trb2nonmult}, we call such a quadruple ``valid'' 
if the coefficients  $a_{I+K},a_{I+L},a_{J+K},a_{J+L}$ are all nonzero.
Let $\cal V$ be the set of valid quadruples.
Trying to mimic the proof of Lemma~\ref{trb2nonmult}, we will write 
\begin{equation} \label{sumoverU}
\tr(B^2)=\sum_{(P,Q,R,I) \in {\cal U}} a_{P}.a_{Q}.a_{R}.a_{Q+R-P}
\end{equation}
where $\cal U$ is the set of quadruples $(P,Q,R,I)$ such that:
\begin{enumerate}
\item $P,Q,R$ and $Q+R-P$ belong to $\cal M$ (the set of monomials of $f$) and $I$ is a row index of $M$, i.e., $I \in \{0,1\}^n$ and $|I|=k$.
\item 
$I \leq P$ and $I \leq Q$.
\item There exists a (unique) row index $J$ 
such that 
$P-I = R-J$.
\end{enumerate}
Equation~(\ref{sumoverU}) follows from~(\ref{sumoverV}) due to the following one-to-one correspondence between quadruples of  ${\cal U}$ and $\cal V$:
\begin{itemize}
\item[(i)] Given a quadruple $(I,J,K,L) \in {\cal V}$, set $P=I+K$, $Q=I+L$, $R=J+K$.
The quadruple $(P,Q,R,I)$ is in $\cal U$ since 
$Q+R-P=J+L$ and $P-I=R-J=K$.
\item[(ii)] A quadruple $(P,Q,R,I) \in {\cal U}$ has a unique preimage 
$(I,J,K,L) \in {\cal V}$, which 
is obtained as follows. A preimage must satisfy $K=P-I$, $L=Q-I$, $J=R-K$. 
This defines a quadruple $(I,J,K,L)$ such that $P-I=R-J$, so $J$ must be a row index by condition~3 in the definition of $\cal U$: $J \in \{0,1\}^n$ 
and $|J|=k$. It follows that
$(I,J,K,L) \in {\cal V}$ and that this quadruple is indeed a preimage of 
$(P,Q,R,I)$.
\end{itemize}
Since $2a_{P}.a_{Q}.a_{R}.a_{Q+R-P} \leq (a_P.a_R)^2+(a_Q.a_{Q+R-P})^2$,
it follows from~(\ref{sumoverU}) that
\begin{equation} \label{2sums}
2\tr(B^2) \leq \sum_{(P,Q,R,I) \in {\cal U}} a_{P}^2.a_{R}^2
+\sum_{(P,Q,R,I) \in {\cal U}} a_{Q}^2.a_{Q+R-P}^2.
\end{equation}
The first sum is upper bounded by 
\begin{equation} \label{sum1}
|{\cal M}|.\left(\sum_{P \in {\cal M}} {\sup(P) \choose k}a_P^2\right).
\left(\sum_{R \in {\cal M}} a_R^2\right)
\end{equation}
since there are at most $|\cal M|$ choices for $Q$ and we must have $I \leq P$ for a quadruple in $\cal U$.
This is equal to $$|{\cal M}|\tr(B)
\left(\sum_{R \in {\cal M}} a_R^2\right)$$
by Lemma~\ref{trbreal}.
Likewise, the second sum in~(\ref{2sums}) is upper bounded 
by $$\sum_{P,Q,R \in {\cal M}} {\sup(Q) \choose k}a_Q^2a_{Q+R-P}^2$$ 
since we have $I \leq Q$ for a quadruple in $\cal U$.
For any fixed $Q \in \cal M$, $\sum_{P,R \in \cal M} a_{Q+R-P}^2 
\leq |{\cal M}|.\sum_{S \in M} a_S^2$ since each term $a_S^2$ on the right-hand side can appear at
 most $|{\cal M}|$ times on the left-hand side.
We conclude that the second sum in~(\ref{2sums}) admits the same upper bound~(\ref{sum1}) 
as the first sum, and the lemma is proved.
\end{proof}

\subsection{Polynomial-time computable lower bounds} \label{polytime}

The lower bound 
$$L(f)=\frac{\sum_{P \in {\cal M}} {\sup(P) \choose k}a_P^2}{|{\cal M}|\sum_{P \in {\cal M}} a_P^2}$$
in Theorem~\ref{realth} is clearly computable in polynomial time 
from $f$ and $k$. 
Recall that we have obtained this lower bound by constructing a symmetric matrix $B$ such that $\dim \partial^{=k} f \geq \tr(B)^2 / \tr(B^2) \geq L(f)$.
The quantity $\tr(B)^2 / \tr(B^2)$ is therefore a better lower bound on 
$\dim \partial^{=k} f$ than $L(f)$. Like $L(f)$, 
it turns out to be computable in polynomial time.
This is not self-evident because $B$ may be of exponential size
(which is the source of the $\sharpp$-hardness result in the next section).
\begin{theorem}
There is an algorithm which, given $f$ and $k$, computes
the lower bound $\tr(B)^2 / \tr(B^2)$ on $\dim \partial^{=k} f$
in polynomial time.
\end{theorem}
\begin{proof}
We build on the proof of Theorem~\ref{realth}. Lemma~\ref{trbreal} shows that 
$\tr(B)$ can be computed in polynomial time, so it remains to do the same
for $\tr(B^2)$. In the proof of Lemma~\ref{trb2real}, we have defined a 
set of quadruples $\cal U$ such that 
$$\tr(B^2)=\sum_{(P,Q,R,I) \in {\cal U}} a_{P}.a_{Q}.a_{R}.a_{Q+R-P}.$$
This can be rewritten as:
$$\tr(B^2)=\sum_{(P,Q,R) \in {\cal M}} N(P,Q,R).a_{P}.a_{Q}.a_{R}.a_{Q+R-P}$$
where we denote by $N(P,Q,R)$ the number of row indices $I$ such that 
$(P,Q,R,I) \in {\cal U}$.
It therefore remains to show that $N(P,Q,R)$ can be computed in polynomial time.
Toward this goal we make two observations.
\begin{itemize}
\item[(i)] Condition~2 in the definition of $\cal U$ means that $I \leq \min(P,Q)$, where the $n$-tuple $\min(P,Q)$ is the coordinatewise minimum of $P$ and
$Q$. 
\item[(ii)] The equality $P-I=R-J$ in condition~3 is equivalent to $P-R=I-J$, 
hence $P-R \in \{-1,0,1\}^n$ since $I,J \in \{0,1\}^n$. Moreover, since 
$I$ and $J$ each have $k$ nonzero coordinates, $P-R$ must contain the
same number of $1$'s and $-1$'s. By observation~(i), 
the positions of $1$'s must be positive in $\min(P,Q)$.

\end{itemize}
We can therefore compute $N(P,Q,R)$ as follows.
\begin{enumerate}
\item If $Q+R-P$ is not a monomial of $f$, $N(P,Q,R)=0$.
\item If $P-R$ is not in $\{-1,0,1\}^n$, $N(P,Q,R)=0$.
\item If  $P-R$ does not contain the same number of $1$'s and $-1$'s, 
$N(P,Q,R)=0$.
\item If some of the positions of $1$'s in $P-R$ contain a $0$ in $\min(P,Q)$, 
$N(P,Q,R)=0$.
\item Let $\mathrm{ones}(P,R)$ be the number of
$1$'s in $P-R$ and $\mathrm{zeros}(P,Q,R)$ the number of $0$'s in $P-R$ 
such that we have a positive entry in $\min(P,Q)$ at the same position.
Then $N(P,Q,R)={\mathrm{zeros}(P,Q,R) \choose k-\mathrm{ones}(P,R)}$: from the relation $P-R=I-J$, the positions with a 1 in $P-R$ must contain a $1$ in $I$.
So it remains to choose the remaining $k-\mathrm{ones}(P,R)$ nonzero positions of $I$. We can only choose them among the positions that are positive in
$\min(P,Q)$ (by (i)) and contain a 0 in $P-R$.
\end{enumerate}
\end{proof}

\subsection{Elementary symmetric polynomials} \label{elementarypoly}

A natural question is whether the inequality $\rk (B) \geq \frac{(\tr B)^2}{\tr(B^2)}$ in Lemma \ref{alon} 
is tight when $B = M^T.M$ and $M$ comes from a partial derivatives matrix.
It is well known that this inequality 
is in general far from tight, since it is obtained by means of the Cauchy-Schwarz inequality.
However in our case, due to the particular shape of the matrix $B$, it is not 
a priori clear whether 
a large gap can exist between $\rk (B)$ and $\frac{(\tr B)^2}{\tr(B^2)}$.
In the following, we show that arbitrarily large gaps can indeed be achieved. Our source of examples are the elementary symmetric polynomials $Sym_{d,n}(x_1, \dots, x_n) = \sum_{|I| = d} x^I$. 
Here the vector $I$ of exponents belongs to $\{0,1\}^n$, and $x^I$ denotes as usual the multilinear monomial $x_1^{i_1}. \cdots . x_n^{i_n}$.

More precisely, we will show the following.
\begin{proposition}
	For any fixed positive integers $d, k < d$, the family of polynomials $f_n = Sym_{d, n}$ has the following property:
		if we consider $u_n = {\rk}(B_n) = \dim \partial^{= k} f_n$ the sequence of dimensions of partial derivatives, and $v_n = \frac{(\tr B_n)^2}{\tr(B_n^2)}$ the sequence of lower bounds for the dimension, we have that $v_n \rightarrow 1$,  whereas $B_n$ is of full rank and, hence, $u_n \rightarrow + \infty$.
\end{proposition}
Note that since $d$ is fixed, the polynomial $f_n$ is sparse: it contains only $n^{O(1)}$ monomials.
\begin{proof}
	The matrix $M$ of partial derivatives of $f_n$ has only 0/1-coefficients and the coefficient $M_{I, J}$ is non-zero iff
	$I \cap J = \emptyset$ (with $|I| = k, |J| = d - k$). This matrix is commonly known as the
	{\it disjointness matrix} and 
has proved useful in  communication complexity~\cite{KN97} and of course in algebraic complexity~\cite{NW96} for the study of elementary symmetric polynomials.\footnote{Variations on this matrix also proved
useful for the analysis of the {\em shifted} partial derivatives of symmetric polynomials~\cite{fournier15}.}
In particular, by \cite{Gottlieb}, we have that $M$ is of full rank,
	i.e., that $u_n = \dim \partial^{=k} f_n = \min\{ \binom{n}{k}, \binom{n}{d-k}\}$.
	This directly implies that $u_n \rightarrow +\infty$.\\
	We already know that $v_n \geq 1$ for all $n$, so we only need to compute an upper bound on $v_n$ that tends to 1 to obtain $v_n \rightarrow 1$.
	To do so, we first compute the coefficients of the matrix $B = M^T.M$:
	\begin{align*}	
		B_{I,J} &= \sum_{|K| = d - k} M_{I,K}M_{J,K} 
					= \sum_{\substack{|K| = d - k \\ K \cap I = K \cap J = \emptyset}} 1 = \binom{n - |I \cup J|}{d - k}
	\end{align*}
	Notice that the value of a diagonal entry $B_{I,I} = \binom{n-k}{d-k}$ is independent of $I$, hence we can easily compute
	the trace of $B$: $\tr(B) = \binom{n - k}{d - k} \binom{n}{k}$.
	A diagonal entry of $B^2$ is of the form $(B^2)_{I,I} = \sum_{|J| = k} (B_{I,J})^2$.
	In order to obtain an upper bound on $v_n$, it is enough to lower bound $\tr(B^2)$.
	Since all the terms are non-negative, we will consider the following subsum:
	\[
		(B^2)_{I,I} \geq \sum_{\substack{|J| = k \\ I \cap J = \emptyset}} (B_{I,J})^2
			= \sum_{\substack{|J| = k \\ I \cap J = \emptyset}} \binom{n-2k}{d-k}^2
			= \binom{n-2k}{d-k}^2 \binom{n-k}{k}.
	\]
	Hence $Tr(B^2) = \sum_{|I| = k}\ (B^2)_{I,I} \geq \binom{n-2k}{d-k}^2 \binom{n-k}{k}  \binom{n}{k}$. Finally, we obtain the following upper bound
	\begin{equation}\label{upboundvn}
		v_n \leq \frac{\binom{n-k}{d-k}^2 \binom{n}{k}^2}{\binom{n-2k}{d-k}^2 \binom{n-k}{k} \binom{n}{k}}
			\xrightarrow[n \rightarrow \infty]{} 1 
	\end{equation}
\end{proof}

This proves that for constant $k,d$, the gap can be as large as we want, but one can ask whether such large gaps can also be achieved 
when $k$ and $d$ are increasing functions of $n$.
Let us consider the case where $k$ and $d$ are proportional to $n$, i.e., $k=\alpha n$ and $d=\beta n$ for some constants $\alpha,\beta <1$.
Now, it is no longer true that $v_n \rightarrow 1$. For example, for $\alpha = 0.2$ and $\beta = 0.4$ we have that $v_n \rightarrow \infty$.
 However, we can still prove that $\frac{v_n}{u_n} \rightarrow 0$ for certain values of $\alpha$ and $\beta$. 
In the following proposition, to make sure that $k=\alpha n$ and $d=\beta n$ are always integers we set $k = k' m , d = d' m$ and $n = n'm$ where $m$ is a new parameter and $k',d',n'$ are constants (so $\alpha=k'/n'$ and $\beta=d'/n'$).
\begin{proposition}\label{badtrace}
	For any positive integers $k', d', n'$ such that $k' < d' < n'/2$, the family of polynomials $f_m = Sym_{d'm, n'm}$ has the following property:
		if we consider $u_m = \dim \partial^{= k' m} f_m$ and $v_m = \frac{(\tr B_m)^2}{\tr(B_m^2)}$, we have $\frac{v_m}{u_m} \rightarrow 0$.
\end{proposition}
\begin{proof}
	We set $k := k' m, d := d' m, n = n'm$. Since $f_m$ has degree $d'm$, then $\dim \partial^{= k' m} f_m = \dim \partial^{= (d'-k') m} f_m $. Hence
	we may assume without loss of generality that $2k' < d'$. Recall that $u_m = \min\{\binom{n}{k},\binom{n}{d-k}\}= \binom{n}{k}$. If we do the same
	proof as in proposition above and use the same upper bound as in (\ref{upboundvn}) on $v_m$, we get:
	\[
		\frac{v_m}{u_m} \leq \frac{\binom{n-k}{d-k}^2}{\binom{n-2k}{d-k}^2 \binom{n-k}{k}} =: w_m
	\]
	For any increasing funcions $g_1(m), g_2(m)$ such that $g_1(m) > g_2(m)$ and $g_1(m) - g_2(m)$ is also increasing, Stirling's approximation formula 
	gives us the following asymptotic for the binomial coefficient $\binom{g_1}{g_2}$: 
		\[ 
			\binom{g_1}{g_2} \sim \sqrt{\frac{g_1}{2\pi g_2 (g_1-g_2)}} \exp[g_1 \log(g_1) - g_2 \log(g_2) - (g_1-g_2) \log(g_1-g_2)]. 
		\]	
	Thus,
	we set $f(t) := t\log(t)$ and obtain that
	\begin{align*}
		w_m \sim P(m) \exp[& f(n-k)  - 2f(n-d) - f(n-2k)  + 2f(n-d-k)  + f(k)  ]
	\end{align*}
	with $P(m)$ the square root of some rational function. We factor out $n$ and operate to obtain
	\vspace*{-0.1cm}
	\[
		w_n \sim P(m) \exp[n(f(1-\alpha) - 2f(1-\beta) - f(1 - 2\alpha) + 2f(1 - \alpha - \beta) + f(\alpha))]
	\]
	with $\alpha := k/n$ and $\beta := d/n$ and $0 < 2 \alpha < \beta < 1/2$. A computer aided computation yields that for
	these values of $\alpha$ and $\beta$, we have that  
	$\gamma := f(1-\alpha) - 2f(1-\beta) - f(1 - 2\alpha) + 2f(1 - \alpha - \beta) + f(\alpha) < 0$, hence we 
	finally obtain that $w_m \sim P(m)\exp[\gamma n' m] \xrightarrow[m \rightarrow \infty]{} 0$.
\end{proof}

\section{$\sharpp$-hardness result for the space of partial derivatives}
\label{hardsection}

In this section it is convenient to work with the space $\partial^+ f$ spanned
by partial derivatives of $f$ of order $r$ 
where $1 \leq r \leq \deg(f)-1$.
We will work with homogeneous polynomials, and for those polynomials we have
$\dim \partial^+ f = \dim \partial^* f - 2$.

\begin{theorem} \label{hardpartials}
It is $\sharpp$-hard to compute $\dim \partial^* f$ for an input polynomial
$f$ given in expanded form (i.e., written as a sum of monomials).
This result remains true for multilinear homogeneous polynomials with
coefficients in $\{0,1\}$. 
\end{theorem}
We proceed by reduction from the problem of 
counting the number of independent sets
in a graph, and use as an intermediate step a problem of topological
origin.
Recall that an (abstract) simplical complex is a family $\Delta$  of subsets
of a finite set $S$ such that for every $F$ in $\Delta$, all the
nonempty subsets of $F$ are also in $\Delta$. The elements of $\Delta$
are also called {\em faces} of the simplicial complex.
We denote by $|\Delta|$ the number of faces of $\Delta$, and more
generally by $|X|$ the cardinality of any finite set $X$.
The {\em dimension} of a face $X \in \Delta$ is $|X|-1$.
The dimension of $\Delta$ is the
maximal dimension of its faces. If every face of $\Delta$ belongs to 
a face of dimension $\dim(\Delta)$, the simplicial complex is said to
be {\em pure}.

The simplicial complex generated by a family $F_1,\ldots,F_m$ of
subsets of $S$ is the smallest simplicial complex containing all of
the $F_i$ as faces. This is simply the family of nomempty subsets 
$Y \subseteq S$ such that $Y \subseteq F_i$ for some $i$.
\begin{theorem} \label{hardcomplex}
The following problem is $\sharpp$-complete: given a family
$F_1,\ldots,F_m$ of subsets 
of $[n]=\{1,\ldots,n\}$, compute the number of faces of the simplical
complex $\Delta$ that it generates. 

This result remains true if $\Delta$ is pure, i.e., if
$F_1,\ldots,F_m$ have the same cardinality.
\end{theorem}
We deduce Theorem~\ref{hardpartials} from
Theorem~\ref{hardcomplex}. Let $\Delta$ be the pure simplicial complex
generated by a family $F_1,\ldots,F_m$ of subsets of $[n]$, with
$|F_i|=d$ for all~$i$.
We associate to each $F_i$ the monomial $m_i=\prod_{j \in F_i} X_j$,
and to $\Delta$ the polynomial
$f(X_1,\ldots,X_n,Y_1,\ldots,Y_m) = \sum_{i=1}^m
Y_i.m_i(X_1,\ldots,X_n)$. 
This is a multilinear homogeneous polynomial of degree $d+1$ in $m+n$ variables.
Theorem~\ref{hardpartials} is an immediate consequence of
Theorem~\ref{hardcomplex} and of the
following lemma.
\begin{lemma}
A basis of the linear space spanned by $\partial^+ f$ consists of
the following set of $2|\Delta|$ polynomials:
\begin{itemize}
\item[(i)] The $|\Delta|$ monomials of the form $\prod_{j \in F} X_j$,
  where $F$ is a face of $\Delta$.

\item[(ii)] The $|\Delta|$ polynomials of the form $\partial f
  / \partial F$, where $F$ is a face of $\Delta$ (we denote by 
$\partial f / \partial F$ the polynomial obtained from $f$ by
differentiating with respect to all variables $X_j$ with $j \in F$).
\end{itemize}
In particular, $\dim(\partial^+ f)=2|\Delta|$.
\end{lemma}
\begin{proof}
We first note that a polynomial in (ii) belongs to $\partial^+ f$ by
definition. A polynomial in (i) also belongs to $\partial^+ f$ since
it can be obtained by picking a maximal face $F_i$ containing $F$,
differentiating with respect to $Y_i$, and then with respect to all
variables $X_j$ where $j \in F_i \setminus F$.

Conversely, any partial derivative which is not identically 0 is of
the form~(i) if we have differentiated $f$ with respect to exactly one
$Y_i$, or of the form~(ii) if we have not differentiated $f$ with
respect to any of the variables $Y_i$. It therefore remains to show
that the polynomials in our purported basis are linearly independent.

The monomials  in (i) are linearly independent since they are pairwise
distinct. To show that the polynomials in (ii) are linearly
independent, consider a linear combination 
$$g=\sum_{j=1}^{|\Delta|} \alpha_j \frac{\partial f}{\partial G_j},$$
where $G_1,\ldots,G_{|\Delta|}$ are the faces of $\Delta$.
By construction of $f$, 
\begin{equation} \label{lincomb}
g=\sum_{i=1}^m Y_i.\left( \sum_{j=1}^{|\Delta|} \alpha_j 
\frac{\partial m_i}{\partial G_j}\right).
\end{equation}
Assume that some coefficient $\alpha_j$, for instance $\alpha_1$, is
different from 0. The face $G_1$ belongs to some maximal face of
$\Delta$, for instance to $F_1$.
We claim that 
$\sum_{j=1}^{|\Delta|} \alpha_j \frac{\partial m_1}{\partial G_j} \neq
0$.
Indeed, the faces $G_j$ which are not included in $F_1$ contribute
nothing to this sum, and the faces that are included in $F_1$
contribute pairwise distinct monomials. It follows
from~(\ref{lincomb}) that $g \neq 0$, and that the polynomials 
in (ii) are indeed linearly independent.

To complete the proof of the lemma, it remains to note that the spaces
spanned by (i) and (ii) are in direct sum. Indeed, the first space is
included in $\qq[X_1,\ldots,X_n]$ while the second is included in 
$\sum_{i=1}^m Y_i\qq[X_1,\ldots,X_n]$.
\end{proof}

\subsection{Proof of Theorem~\ref{hardcomplex}}

Let $G=(V,E)$ be a graph with vertex set $V=[n]$ and $m=|E|$ edges. 
We associate to $G$ the
simplicial complex $\Delta$ generated by the complements of the edges
of $G$, i.e., by the sets $V \setminus \{u,v\}$ where $uv \in E$.
This is a pure simplicial complex of dimension $n-2$.
The faces of $\Delta$ are the complements of the dependent sets of $G$, except the empty set (the complement of $V$) which is not a face of $\Delta$ by convention.
Hence $|\Delta|=2^n - \Ind(G) - 1$, where $\Ind(G)$ denotes the number of
independents sets in $G$. Computing the number of independent sets of
a graph is a well-known $\sharpp$-complete problem. 
It was shown to be  $\sharpp$-complete even 
for bipartite graphs~\cite{provan83},
for planar bipartite graphs of degree at most four~\cite{vadhan01} and for 3-regular graphs~\cite{greenhill00}.
It follows that
computing $|\Delta|$ is $\sharpp$-hard, and membership in $\sharpp$ is
immediate from the definition. 
This completes the proof of Theorem~\ref{hardcomplex}.

We note that it is easy to shortcircuit Theorem~\ref{hardcomplex} and
construct the polynomial $f$ in the proof of
Theorem~\ref{hardpartials} directly from $G$: 
we have $$f=\sum_{uv \in E} Y_{uv}.\prod_{w {\not \in} \{u,v\}} X_w$$
and $\dim \partial^+ f = 2(2^n-\Ind(G)-1)$.
Since the maximal faces of $\Delta$ 
have $n-2$ elements, we have
the following refinement of Theorem~\ref{hardpartials}.
\begin{corollary} 
It is $\sharpp$-hard to compute $\dim \partial^* f$ for a multilinear
homogenous polynomial $f$ of degree $n-1$ with coefficients in $\{0,1\}$,
$m$ monomials and $n+m$ variables with $m \leq {n \choose 2}$.
\end{corollary}

\small
\bibliographystyle{plain}

\end{document}